\documentclass[copyright]{eptcs}
\providecommand{\event}{Express 2011} 
\usepackage{breakurl}             

\usepackage{amsmath,bm,mathpartir,amssymb}
\usepackage{color}
\usepackage[all]{xy}
\definecolor{g3}{gray}{0.0}
\definecolor{g2}{gray}{0.2}
\definecolor{g1}{gray}{0.5}

\newif\iflong\longfalse

\newcommand{\afficheshort}[1]{}
\newcommand{\onlyshort}[1]{
  \renewcommand{\afficheshort}{#1}
  \iflong
  \renewcommand{\afficheshort}{}
  \fi
  \afficheshort}

\newcommand{\encod}[1]{\ensuremath{[\![#1]\!]}}

\newcommand{\set}[1]{\ensuremath{\{#1\}}}
\newcommand{\OR}{\ensuremath{~\big|~}}
\newcommand{\defgram}{\ensuremath{~~::=~~}}
\newcommand{\eqdef}{\ensuremath{\stackrel{def}{=}}}
\renewcommand{\max}[1]{\ensuremath{\mathrm{max}(#1)}}
\renewcommand{\min}[1]{\ensuremath{\mathrm{min}(#1)}}
\newcommand{\rar}{\ensuremath{\longrightarrow}}

\newcommand{\new}{\ensuremath{\bm{\nu}}}

\newcommand{\nil}{\ensuremath{\bm{0}}}
\newcommand{\msg}[1]{\ensuremath{\langle#1\rangle}}
\newcommand{\out}[1]{\ensuremath{\overline#1}}
\newcommand{\outm}[2]{\ensuremath{\overline#1\msg{#2}}}

\newcommand{\fnames}[1]{\ensuremath{\mathrm{fn}(#1)}}
\newcommand{\bnames}[1]{\ensuremath{\mathrm{bn}(#1)}}
\newcommand{\names}[1]{\ensuremath{\mathrm{names}(#1)}}
\newcommand{\rcvnames}[1]{\ensuremath{\mathrm{rcv}(#1)}}
\newcommand{\resnames}[1]{\ensuremath{\mathrm{res}(#1)}}
\newcommand{\father}[1]{\ensuremath{\mathrm{father}(#1)}}
\newcommand{\son}[1]{\ensuremath{\mathrm{son}(#1)}}

\newcommand{\xr}[1]{\ensuremath{\xrightarrow{#1}}}
\newcommand{\red}{\ensuremath{\longrightarrow}}

\newcommand{\chan}[2]{\ensuremath{\sharp^{#1}#2}}
\newcommand{\ichan}[2]{\ensuremath{\mathsf{i}^{#1}#2}}
\newcommand{\otypename}{\ensuremath{\mathsf{o}}}
\newcommand{\ochan}[2]{\ensuremath{\otypename^{#1}#2}}

\newcommand{\topsign}[2]{\ensuremath{\uparrow}}
\newcommand{\tchan}[2]{\ensuremath{\topsign\!\!#2}}

\newcommand{\os}[1]{\ensuremath{\mathcal{M}(#1)}}
\newcommand{\osG}[2]{\ensuremath{\mathcal{M}_{#1}(#2)}}
\newcommand{\rname}[1]{}
\iflong
\renewcommand{\rname}[1]{\textsc{#1}}
\fi

\newcommand{\Unit}{\ensuremath{\mathbb{U}}}
\newcommand{\unit}{\ensuremath{\star}}
\newcommand{\lvl}[1]{\ensuremath{\mathsf{lvl}(#1)}}
\newcommand{\Gam}{\ensuremath{\Gamma}}
\newcommand{\typ}[1]{\ensuremath{\vdash#1{\tiny \textcolor{red}{\ensuremath{^{OLD}}}}}}
\newcommand{\typj}[2]{\ensuremath{#1\,\vdash\,#2}}
\newcommand{\typlpi}[2]{\ensuremath{#1\,\vdash^{\tiny L\pi}\,#2}}
\newcommand{\entailinf}{\ensuremath{\vdash^{\tiny L\pi}_{\mathtt{TI}}}}
\newcommand{\typinf}[2]{\ensuremath{#1\,\entailinf\,#2}}
\newcommand{\typdeng}[2]{\ensuremath{#1\,\vdash_{D}\,#2}}

\newcommand{\subt}[2]{\ensuremath{#1\leq#2}}

\newcommand{\dom}[1]{\ensuremath{\mathrm{dom}(#1)}}

\newcommand{\deriv}{\ensuremath{\mathcal{D}}}
\newcommand{\Lpi}{L$\pi$}
\newcommand{\STlam}{ST$\lambda$}

\newtheorem{defi}{Definition}
\newtheorem{thm}[defi]{Theorem}

\newtheorem{prop}[defi]{Proposition}
\newtheorem{lem}[defi]{Lemma}

\newtheorem{rk}[defi]{Remark}
\newtheorem{expl}[defi]{Example}
\newcommand{\qed}{\ensuremath{\Box}}
\newenvironment{proof}{\noindent\textbf{Proof}\begin{sl}}{\hfill\qed\end{sl}}
\newenvironment{proofsketch}{\noindent\textbf{Proof (sketch).}\begin{sl}}{\end{sl}}

\newcommand{\daniel}[1]{\textcolor{blue}{#1}}

\title{Termination in a $\pi$-calculus with Subtyping}
\author{
Ioana Cristescu 
\qquad\qquad
Daniel Hirschkoff
\institute{ENS Lyon, Universit\'e de Lyon, CNRS, INRIA, France}
}
\date{}
\def\titlerunning{Termination in a $\pi$-calculus with Subtyping}
\def\authorrunning{I.~Cristescu, D.~Hirschkoff}

\begin{document}

\maketitle
\begin{abstract}
We present a type system to guarantee termination of
  $\pi$-calculus processes that exploits input/output capabilities and
  subtyping, as originally introduced by Pierce and Sangiorgi, in
  order to analyse the usage of channels.

  We show that our system improves over previously existing proposals
  by accepting more processes as terminating. This increased
  expressiveness allows us to capture sensible programming idioms.
  We demonstrate how our system can be extended to handle the
  encoding of the simply typed $\lambda$-calculus, and
  discuss questions related to type inference.



\end{abstract}

\section{Introduction}
\label{sec:intro}
Although many concurrent systems, such as servers, are supposed to run
forever, termination is an important property in a concurrent
setting. For instance, one would like a request to a server to be eventually
answered; similarly, the access to a shared resource should be
eventually granted.  Termination can be useful to guarantee in turn
lock-freedom properties~\cite{DBLP:journals/toplas/KobayashiS10}.

In this work, we study termination in the setting of the
$\pi$-calculus: concurrent systems are specified as $\pi$-calculus
processes, and we would like to avoid situations in which a process
can perform an infinite sequence of internal communication
steps. Despite its conciseness, the $\pi$-calculus can express complex
behaviours, such as reconfiguration of communication topology, and
dynamic creation of channels and threads. Guaranteeing termination is
thus a nontrivial task.
%
%

More specifically, we are interested in methods that provide
 termination guarantees statically.
There exist several type-based approaches to guarantee termination in
the
$\pi$-calculus~\cite{deng:sangiorgi:termination:IC,yoshida:berger:honda:termination:ic,sangiorgi:mscs:termination,DBLP:conf/birthday/DemangeonHS09,demangeon:hirschkoff:sangiorgi:concur10}. In
these works, any typable process is guaranteed to be reactive, in the
sense that it cannot enter an infinite sequence of internal
communications: it eventually terminates computation, or ends up in a
state where an interaction with the environment is required.

The type systems in the works mentioned above have different
expressive powers. Analysing the expressiveness of a type system for
termination amounts to studying the class of processes that are
recognised as terminating. A type system for termination typically
rules out some terminating terms, because it is not able to recognise
them as such (by essence, an effective type system for termination
defines an approximation of this undecidable property).
When improving expressiveness, one is interested in making the type
system more
flexible: 
more processes should be deemed as terminating. An important point in
doing so is also to make sure that (at least some of) the `extra
processes' make sense from the point of view of programming.


\paragraph{Type systems for termination in the $\pi$-calculus.}

Existing type systems for termination in the $\pi$-calculus build on
simple types~\cite{SW01}, whereby the type of a channel describes what
kind of values it can carry. Two approaches, 
that we shall call `level-based' and
`semantics-based', have been studied to guarantee termination of
processes. 
We discuss below the first kind of methods, and return to
semantics-based approaches towards the end of this section.
Level-based methods for the termination of processes originate
in~\cite{deng:sangiorgi:termination:IC}, and have been further
analysed and developed
in~\cite{DBLP:conf/birthday/DemangeonHS09}.
They exploit a stratification of names, obtained by associating a \emph{level}
(given by a natural number) to each name. Levels are used to insure
that at every reduction step of a given process, some well-founded
measure defined on  processes decreases.

\smallskip

Let us illustrate the level-based approach on some examples. In this
paper, we work in the asynchronous $\pi$-calculus, and replication can
occur only on input prefixes. As  in previous work, adding
features like synchrony or the sum operator to our setting does not
bring any difficulty. 

According to level-based type systems, the process $!a(x).\outm{b}x$
is well-typed provided \lvl{a}, the level of $a$,
is strictly greater than \lvl{b}. Intuitively, this process
trades messages on $a$ (that `cost' \lvl a) for messages on $b$ (that
cost less).
Similarly, $!a(x).(\outm b x\,|\, \outm b x)$ is also well-typed,
because none of the two messages emitted on $b$ will be liable to
trigger messages on $a$ ad infinitum. More generally, for a process of
the form $!a(x).P$ to be typable, we must check that all messages
occurring in $P$ are transmitted on channels whose level is strictly
smaller than \lvl a (more accurately, we only take into account those
outputs that do not occur under a replication in $P$ --- see
Section~\ref{section:io}).

This approach rules out a process like $!a(x).\outm b x ~|~
!b(y).\outm a y$ (which generates the unsatisfiable constraint $\lvl
a>\lvl b>\lvl a$), as well as the other obviously `dangerous' term
$!a(x).\outm a x$ --- note that neither of these processes is
diverging, but they lead to infinite computations as soon as they are
put in parallel with a message on $a$.


\paragraph{The limitations of simple types.}

The starting point of this work is the observation that since
existing level-based systems rely on simple types, they rule out
processes that are harmless from the point of view of termination,
essentially because in simple types, all names transmitted on a given
channel should have the same type, and hence, in our setting, the same
level as well.

If we try for instance to type the process $P_0 \eqdef !a(x).\outm x
t$, the constraint is $\lvl a>\lvl x$, in other words, the level of
the names transmitted on $a$ must be smaller than $a$'s level. 
%
It should therefore be licit to put $P_0$ in parallel with $\outm a
p\,|\,\outm a q$, provided $\lvl p<\lvl a$ and $\lvl q<\lvl
a$. Existing type systems enforce that $p$ and $q$ have \emph{the same
  type} for this process to be typable: as soon as two names are sent
on the same channel (here, $a$), their types are unified. This means
that if for some reason (for instance, if the subterm $!p(z).\outm q
z$ occurs in parallel) we must have $\lvl p>\lvl q$, 
the resulting process is rejected, although it is terminating.

We would like to provide more flexibility in the handling of the level
of names, by relaxing the constraint that $p$ and $q$ from the example
above should have the same type. To do this while preserving soundness
of the type system, 
it is necessary to take into account the way names are used in the
continuation of a replicated input. In process $P_0$ above, 
 $x$ is used in output in the continuation, which allows one
to send on $a$ any name (of the appropriate simple type) of level
\emph{strictly smaller} than \lvl a. If, on the other hand, we
consider process $P_1 \eqdef !b(y).!y(z).\outm c z$, then typability
of the subterm $!y(z).\outm c z$ imposes $\lvl y>\lvl c$, which means
that any name of level \emph{strictly greater} than \lvl c can be sent
on $b$. In this case, $P_1$ uses the name
$y$ that is received along $b$ in input.
%
We can remark that divergent behaviours would arise if we allowed the
reception of names having a bigger (resp.\ smaller) level in $P_0$
(resp.\ $P_1$).

\paragraph{Contributions of this work.}


These observations lead us to introduce a new type system for
termination of mobile processes based on Pierce and Sangiorgi's
system for \emph{input/output types}
(i/o-types)~\cite{DBLP:journals/mscs/PierceS96}.  
I/o-types are based on the
notion of \emph{capability} associated to a channel name, which makes
it possible to grant only the possibility
of emitting (the output capability) or receiving (the input
capability) on a given channel.  
A subtyping relation is introduced to express the fact that a channel
for which both capabilities are available can be coerced to a channel
where only one is used.
Intuitively, being able to
have a more precise description of how a name will be used can help in
asserting termination of a process: in $P_0$, only the output
capability on $x$ is used, which makes it possible to send a name of
smaller level on $a$; in $P_1$, symmetrically, $y$ can have a bigger
level than expected, as only the input capability on $y$ is
transmitted.

The overall setting of this work is presented in
Section~\ref{section:io}, together with the definition of our type
system. 
This system is strictly more expressive than previously existing
level-based systems. We show in particular that our approach yields a
form of `level polymorphism', which can be interesting in terms of
programming, by making it possible to send several requests to a given
\emph{server} (represented as a process of the form $!f(x).P$, which
corresponds to the typical idiom for functions or servers in
the $\pi$-calculus) with arguments that must have different levels,
because of existing dependencies between them.

\smallskip

In order to study more precisely the possibility to handle terminating
functions (or servers) in our setting, we analyse an encoding of the
$\lambda$-calculus in the $\pi$-calculus. We have presented
in~\cite{DBLP:conf/birthday/DemangeonHS09} a counterexample showing
that existing level-based approaches are not able to recognise as
terminating the image of the simply-typed $\lambda$-calculus (\STlam)
in the $\pi$-calculus (all processes computed using such an encoding
terminate~\cite{SW01}). We show that this counterexample is typable in
our system, but we exhibit a new counterexample, which is not. This
shows that despite the increased expressiveness, level-based methods
for the termination of $\pi$-calculus processes fail to capture
terminating sequential computation as expressed in \STlam.


To accommodate functional computation, we exploit the work presented
in~\cite{demangeon:hirschkoff:sangiorgi:concur10}, where an
\emph{impure} $\pi$-calculus is studied. Impure means here that one
distinguishes between two kinds of names. On one hand,
\emph{functional names} are subject to a certain discipline in their
usage, which intuitively arises from the way names are used in the
encoding of \STlam{} in the $\pi$-calculus. On the other hand,
\emph{imperative names} do not obey such conditions, and are called so
because they may lead to forms of stateful computation (for instance,
an input on a certain name is available at some point, but not later,
or it is always available, but leads to different computations at
different points in the execution).

In~\cite{demangeon:hirschkoff:sangiorgi:concur10}, termination is
guaranteed in an impure $\pi$-calculus by using a level-based approach
for imperative names, while functional names are dealt with
separately, using a semantics-based
approach~\cite{yoshida:berger:honda:termination:ic,sangiorgi:mscs:termination}.
We show that that type system, which combines both approaches for
termination in the $\pi$-calculus, can be revisited in our setting. We
also demonstrate that the resulting system improves in terms of
expressiveness over~\cite{demangeon:hirschkoff:sangiorgi:concur10},
from several points of view.

\medskip

Several technical aspects in the definition of our type systems are
new with respect to previous works. First of all, while the works we
rely on for termination adopt a presentation {\`a} la Church, where
every name has a given type a priori, we define our systems {\`a} la
Curry, in order to follow the approach for i/o-types
in~\cite{DBLP:journals/mscs/PierceS96}. As we discuss below, this has
some consequences on the soundness proof of our systems.
%
Another difference is in the presentation of the impure calculus:
\cite{demangeon:hirschkoff:sangiorgi:concur10} uses a specific
syntactical construction, called \texttt{def}, and akin to a
\texttt{let\,..\,in} construct, to handle functional names. By a
refinement of i/o-types, we are able instead to enforce the discipline
of functional names without resorting to a particular syntactical
construct, which allows us to keep a uniform syntax.

\medskip

We finally discuss 
type inference, by focusing on the
case of the \emph{localised $\pi$-calculus} (\Lpi). \Lpi{} corresponds
to a certain restriction on i/o-types. This restriction is commonly
adopted in implementations of the $\pi$-calculus. We describe a sound
and complete type inference procedure for our level-based
system 
in \Lpi.
We also provide some remarks about inference for i/o-types in the
general case.

\paragraph{Paper outline.}

Section~\ref{section:io} presents our type system, and shows that it
guarantees termination. We study its expressiveness in
Section~\ref{sec:expr}. 
Section~\ref{sec:inference} discusses type
inference, and we give concluding remarks in Section~\ref{sec:concl}.
For lack of space, several proofs are omitted from this version of the
paper.

\section{A Type System for Termination with Subtyping}
\label{section:io}

\subsection{Definition of the Type System}

\paragraph{Processes and types.}

We work with an infinite set of \emph{names}, ranged over using
$a,b,c,\dots,x,y,\dots$. Processes, ranged over using $P, Q, R,\dots$,
are defined by the following grammar (\unit{} is a constant, and we
use $v$ for values):
\begin{mathpar}
  P
  \defgram
  \nil \OR P_1|P_2 \OR \outm a v\OR (\new a)\,P \OR a(x).P\OR !a(x).P
  \and
  v \defgram \unit\OR a
  \enspace.
\end{mathpar}

The constructs of restriction and (possibly replicated) input are
binding, and give rise to the usual notion of $\alpha$-conversion. We
write \fnames{P} for the set of free names of process $P$, and
$P[b/x]$ stands for the process obtained by applying the
capture-avoiding substitution of $x$ with $b$ in $P$.

We moreover implicitly assume, in the remainder of the paper, that all
the processes we manipulate are written in such a way that all bound
names are pairwise distinct and are different from the free names. 
This may in particular involve some implicit renaming
of processes when a reduction is performed.

The grammar of types is given by:
\begin{mathpar}
  T
  \defgram
  \chan{k}{T}\OR \ichan{k}{T} \OR \ochan{k}{T}\OR\Unit
  \enspace,
\end{mathpar}
\noindent where $k$ is a natural number that we call a \emph{level},
and \Unit{} stands for the \texttt{unit} type having \unit{} as only
value. 
A name having type \chan k T has level $k$, and can be used to send or
receive values of type $T$, while type \ichan k T (resp.\ \ochan k T)
corresponds to having only the input (resp.\ output) capability.

Figure~\ref{fig:types:io} introduces the subtyping and typing
relations. We note $\leq$ both for the subtyping relation and for the
inequality between levels, as no ambiguity is possible. 
We can remark that the input (resp.\ output) capability is covariant
(resp.\ contravariant) w.r.t.\ $\leq$, but that the opposite holds for
levels: input requires the supertype to have a smaller level.

\Gam{} ranges over typing environments, which are partial maps from
names to types -- we write $\Gam(a)=T$ if \Gam{} maps $a$ to
$T$. \dom{\Gam}, the domain of \Gam, is the set of names for which
\Gam{} is defined, and $\Gam, a:T$ stands for the typing environment
obtained by extending \Gam{} with the mapping from $a$ to $T$, this
operation being defined only when $a\notin\dom{\Gam}$.

The typing judgement for processes is of the form \typj{\Gam}{P:w},
where $w$ is a natural number called the \emph{weight} of $P$.
The weight corresponds to an upper bound on the maximum level of a
channel that is used 
in output in $P$, without this output occurring under a
replication. This can be read from the typing rule for output messages
(notice that in the first premise, we require the output capability on
$a$, which may involve the use of subtyping) and for parallel
composition.
As can be seen by the corresponding rules, non replicated input prefix
and restriction do not change the weight of a process.  The weight is
controlled in the rule for replicated inputs, where we require that
the level of the name used in input is strictly bigger than the weight
of the continuation process. 
%
%
 We can also observe that working in a synchronous calculus would
 involve a minor change: typing a synchronous output $\outm a v.P$
 would be done essentially like typing $\outm a v\,|\,P$ in our setting
 (with no major modification in the correctness proof for our type
 system).

As an abbreviation, we shall omit the content of messages in prefixes,
and write $a$ and \out a for $a(x)$ and \outm{ a}{ \unit} respectively,
when $a$'s type indicates that $a$ is used to transmit values of type
\Unit.


%


\begin{expl}\label{expl:simple:ex}
  The process $!a(x).\outm x t \,|\, \outm a p\,|\,\outm a q\,|\,
  !p(z).\outm q z$ from Section~\ref{sec:intro} can be typed in our
  type system: we can set 
 $a:\chan{3}{\ochan 2 T},
  p:\chan 2 T, q:\ochan 1 T$. Subtyping on levels is at work in order
  to typecheck the subterm \outm a q.  We provide a more complex term,
  which can be typed using similar ideas, in
  Example~\ref{expl:level-polymorphism} below.
\end{expl}

\begin{figure}[h]
  Subtyping
  \quad
  $\leq$ is the least relation that is reflexive, transitive,
and satisfies the following rules:
  \begin{mathpar}
  \inferrule[\rname{Subt-\#I}]{~}{\chan{k}{T} \leq \ichan{k}{T}}
  \and
  \inferrule[\rname{Subt-\#O}]{~}{\chan{k}{T} \leq \ochan{k}{T}}
  \and
  \inferrule[\rname{Subt-II}]
  {{T \leq S} \and {k_1 \leq k_2}}
  {\ichan{k_2}{T} \leq \ichan{k_1}{S}}
  \and
  \inferrule[\rname{Subt-OO}]
  {{T \leq S}\and {k_1 \leq k_2}}
  {\ochan{k_1}{S} \leq \ochan{k_2}{T}}
\end{mathpar}
Typing values
  \begin{mathpar}
    \inferrule*{~}{\typj{\Gam}{\unit:\Unit}}
    \and
    \inferrule*{\Gam(a) = T}{\typj{\Gam}{a:T}}
    \and
    \inferrule*{\typj{\Gam}{a:T}\and T\leq U
    }{
      \typj{\Gam}{a:U}
    }
    \end{mathpar}
Typing processes
    \begin{mathpar}
    \inferrule[\rname{Nil-S}]{~}{\typj{\Gam}{\nil:0}}
    \and
    \inferrule[\rname{Out-S}]{\typj{\Gam}{a:\ochan k T}\and\typj{\Gam}{v:T}
    }{
      \typj{\Gam}{\outm a v:k}
    }
    \and
    \inferrule[\rname{Inp-S}]{\typj{\Gam}{a:\ichan k T}\and \typj{\Gam,x:T}{P:w}
    }{
      \typj{\Gam}{a(x).P:w}
    }
    \and
    \inferrule[\rname{Rep-S}]{\typj{\Gam}{a:\ichan k T}\and \typj{\Gam,x:T}{P:w}\and
      k>w
    }{
      \typj{\Gam}{!a(x).P:0}
    }
    \and
    \inferrule[\rname{Res-S}]{\typj{\Gam,a:T}{P:w} }{ \typj{\Gam}{(\new a)\,P:w}}
    \and
    \inferrule[\rname{Par-L}]{\typj{\Gam}{P_1:w_1}\and\typj{\Gam}{P_2:w_2}
    }{
      \typj{\Gam}{P_1|P_2:\max{w_1,w_2}}
    }
  \end{mathpar}
  \caption{Typing and Subtyping Rules}
  \label{fig:types:io}
\end{figure}
\begin{figure}[h]
  \begin{mathpar}
    \inferrule{~}{a(x).P~|~ \outm a v\,\red\, P[v/x]}\and
    \inferrule{~}{!a(x).P~|~ \outm a v\,\red\, !a(x).P~|~P[v/x]}
    \\
    \inferrule{P\red P'}{P|Q\red P'|Q}\and
    \inferrule{P\red P'}{(\new a)\,P\red (\new a)\,P'}\and
    \inferrule{Q\equiv P\and P\red P'\and P'\equiv Q'}{Q\red Q'}
  \end{mathpar}
  \caption{Reduction of Processes}
  \label{fig:reduction}
\end{figure}

\paragraph{Reduction and termination.}

The definition of the operational semantics relies on a relation of
structural congruence, noted $\equiv$, which is the smallest
equivalence relation that is a congruence, contains
$\alpha$-conversion, and satisfies the following axioms:
\begin{mathpar}
  P|(Q|R) \,\equiv\, (P|Q)|R
  \and
  P|Q \,\equiv\, Q|P
  \and
  P|\nil\,\equiv\, P
  \\
  (\new a)\,\nil\,\equiv\,\nil
  \and
  (\new a)(\new b)\,P\,\equiv\, (\new b)(\new a)\,P
  \and
  (\new a)\,(P|Q)\,\equiv\, P\,|\,(\new a)\,Q\mbox{ if }a\notin\fnames{P}
\end{mathpar}
Note in particular that there is no structural congruence law
for replication. 

Reduction, written \red, is defined by the rules of
Figure~\ref{fig:reduction}.  

\begin{defi}[Termination]
  A process $P$ \emph{diverges} if there exists an infinite sequence of
  processes $(P_i)_{i\geq 0}$ such that $P=P_0$ and for any $i$,
  $P_i\red P_{i+1}$.
  $P$ \emph{terminates} (or $P$ is terminating) if $P$
  does not diverge.
\end{defi}

\subsection{Properties of the Type System}

We first state some (mostly standard) technical properties satisfied by
our system.

\begin{lem}
  \label{lem:bigger_weight}
  If \typj{\Gam}{P:w} and $w\neq 0$ then for any $w'\geq w$, \typj{\Gam}{P:w'}.
\end{lem}






\begin{prop}[Narrowing]\label{prop:narrowing}
 If \typj{\Gam, x:T}{P:w} and $T' \leq T$,  then \typj{\Gam,
   x:T'}{P:w'} for some $w' \leq w$.

\end{prop}

\begin{lem}\label{lem:subject:congruence}
  If $P\equiv Q$, then \typj{\Gam}{P:w} iff \typj{\Gam}{Q:w}.
\end{lem}

\begin{lem}\label{lem:subject:substitution}
  If \typj{\Gam,x:T}{P:w}, \typj{\Gam}{b:T'} and $T' \leq T$ then
  \typj{\Gam}{P[b/x]:w'}, for some $w'\leq w$.
\end{lem}
\begin{proofsketch}
%
  This is a consequence of Lemma~\ref{prop:narrowing}, as we
  replace $x$ by a name of smaller type.
\end{proofsketch}



\begin{thm}[Subject reduction]\label{thm:subject:reduction}
  If \typj{\Gam}{P:w} and $P\red P'$, then \typj{\Gam}{P':w'} for some
  $w'\leq w$.
\end{thm}

\onlyshort{\begin{proofsketch}}
\iflong
\begin{proof}
\fi
  By induction over the derivation of $P\red P'$. 
\onlyshort{The most interesting case corresponds to the case where }
\iflong
  We only consider the following cases:
  \begin{itemize}
  \item $P\red P'$ with $Q\red Q'$, $P \equiv Q$ and  $P' \equiv Q'$.
    Using induction we derive that \typj{\Gam}{Q:w_q},
    \typj{\Gam}{Q':w_q'} and $w_q' \leq w_q$. We still have to prove that if $P \equiv Q$ (resp. $P' \equiv Q'$) then
\typj{\Gam}{P:w_q} (resp. \typj{\Gam}{P':w_q'}), which can be done by
induction on $P \equiv Q$.

\daniel{j'ai l'impression que cette preuve est ``a l'envers''; en
  plus, il faudrait utiliser le lemme sur structural congruence}

  \item $P = a(x).P_1~|~ \outm a v$ \red $P' = P_1[v/x]$. 

    We will first show that we can derive consistent types for $a$, $x$, $v$ and $P_1$ from $P$'s type and then deduce the type of  $P'$.
  
    From \typj{\Gam}{P:w}, by applying rule \textsc{Par-L}, we deduce that \typj{\Gam}{a(x).P_1 :w_1} and \typj{\Gam}{\outm a v :k} with $w = max(w_1, k)$. 
    On these, we can apply rules \textsc{Inp-S} and \textsc{Out-S} to have \typj{\Gam}{a:\ichan{k}{T}} and \typj{\Gam}{a:\ochan{k}{T'}}.
    Since the two channels have to communicate one with another, we have that  \typj{\Gam}{a:\chan{k}{T''}}, $T' \leq T'' \leq T.$
    This holds with our type system as \typj{\Gam}{a:\ichan{k}{T}} and \typj{\Gam}{a:\ochan{k}{T'}} can be obtained using the rules \textsc{Subt-II} and \textsc{Subt-OO} on $\chan{k}{T''}$.
    
    We use Lemma ~\ref{lem:subject:substitution} on $P_1[v/x]$ obtaining that \typj{\Gam}{P_1[v/x]:w_1'}, $w_1'\leq w_1$. 
    Since $w'=w_1'$ and $w = max(w_1, k)$ we conclude with $w'\leq w$.

  \item
\fi
 $P = !a(x).P_1~|~ \outm a v \red P' = !a(x).P_1~|~ P_1[v/x]$.
 By typability of $P$, we have \typj{\Gam}{!a(x).P_1:0}.
%
 Let $T_a=\Gam(a)$. Typability of
 $P$ gives \typj{\Gam,x:T}{P_1:w_1} for some $T$ and  $w_1$ such that
 $T_a\leq \ichan{k}{T}$ and $w_1<k\leq\lvl{a}$. Typability of \outm
 a v gives $T_a\leq \ochan{k'}{U}$ for some $k'\geq\lvl{a}$, with
 $w=k'$ and \typj{\Gam}{v:U}.
   The two constraints on $T_a$ entail $T\leq U$, and hence, by
   Lemma~\ref{lem:subject:substitution}, \typj{\Gam}{P_1[v/x]:w_2} for
   some $w_2\leq w_1\leq\lvl{a}\leq k'$.
   We then conclude \typj{\Gam}{P':w_2}.
\iflong
  \end{itemize}
\end{proof}
\fi
\onlyshort{\end{proofsketch}}

\newcommand{\termination}{\paragraph{Termination.}}

\termination Soundness of our type system, that is, that every typable
process terminates, is proved by defining a measure on processes that
decreases at each reduction step. A typing judgement
\typj{\Gam}{P:w} yields the weight $w$ of process $P$, but this notion
is not sufficient (for instance, $\out a\,|\,\out a\,|\,a \red
\out a$, and the weight is preserved).
We instead adapt the approach of~\cite{phd:demangeon}, and define the
measure as a multiset of natural numbers. This is done by induction
over the derivation of a typing judgement for the process.  We will use
\deriv{} to range over typing derivations, and write
$\deriv:\typj{\Gam}{P:w}$ to mean that \deriv{} is a derivation of
\typj{\Gam}{P:w}.

To deduce termination, we rely on the multiset extension of the
well-founded order on natural numbers, that we write $>_{mul}$.  $M_2
>_{mul} M_1 $ holds if $M_1 = N \uplus N_1$, $M_2 = N \uplus N_2$, $N$
being the maximal such multiset ($\uplus$ is multiset union), and for
all $e_1 \in N_1$ there is $e_2 \in N_2$ such that $e_1 < e_2$.
The relation  $>_{mul}$ is  well-founded. We write $M_1\geq_{mul}M_2$
if $M_1>_{mul}M_2$ or $M_1=M_2$.


\begin{defi}\label{def:measure}
  Suppose $\deriv:\typj{\Gam}{P:w}$. We define a multiset of natural
  numbers, noted \os{\deriv}, by induction over \deriv{} as follows:
  \begin{mathpar}
    {\mbox{If } \deriv:\typj{\Gam}{\nil}  \mbox{ then
      }  \os{\deriv} = \emptyset} 
    \and
    {\mbox{If } \deriv:\typj{\Gam}{\outm a b:k} \mbox{
        then }  \os{\deriv} = \set{lvl(a)}} 
    \and
    {\mbox{If } \deriv:\typj{\Gam}{!a(x).P:0}  \mbox{
        then } \os{\deriv} = \emptyset} 
    \and
   {\mbox{If } \deriv:\typj{\Gam}{a(x).P:w},
      \mbox{ then }  \os{\deriv} =
     \os{\deriv_1}, \mbox{ where }\deriv_1:\typj{\Gam,x:T}{P:w} }
    \and
    {\mbox{If } \deriv:\typj{\Gam}{(\new
        a)\,P:w},  \mbox{ then } 
      \os{\deriv} =\os{\deriv_1} \mbox{, where }
      \deriv_1:\typj{\Gam, a:T}{P:w}} 

    \and
    {\mbox{If } \deriv: \typj{\Gam}{P_1|P_2:\max{w_1,w_2}}, \mbox{
        then }  \os{\deriv} = \os{\deriv_1} \uplus \os{\deriv_2}
      \mbox{, where } \deriv_i:\typj{\Gam}{P_i}, \,i =1,2
    }
  \end{mathpar}
  Given \Gam{} and $P$, we define \osG{\Gam}{P}, the measure of $P$
  with respect to \Gam, as follows:
  $$
    \osG{\Gam}{P} = \min{\os{\deriv},\, \deriv:{\typj{\Gam}{P:w}}\mbox{
        ~~for some }w}
    \enspace.
  $$
\end{defi}

Note that in the case of output in the above definition, we refer to
\lvl{a}, which is the level of $a$ according to \Gam{} (that is,
without using subtyping).
We have that if \typj{\Gam}{P:w}, then $\forall k\in\osG{\Gam}{P}, k\leq w$.




\begin{lem}
  \label{lem:multiset:subst}
  Suppose \typj{\Gam}{P:w},  
$\Gam(x) = T$, $\Gam(v) = T'$ and $T' \leq T$. Then  
  \osG{\Gam}{P} $\geq_{mul}$ \osG{\Gam}{P[v/x]}. 
\end{lem}
\begin{proof}
  Follows from Lemma \ref{lem:subject:substitution}, and by definition
  of \osG{\Gam}{\cdot}.
\end{proof}

\begin{lem}
\label{lem:congr:os}
If \typj{\Gam}{P:w} and $P \equiv Q$, then \typj{\Gam}{Q:w'} for some
$w'$ and \osG{\Gam} P $=$ \osG{\Gam}Q.
\end{lem}

We are now able to derive the essential property of \osG{\Gam}{\cdot}:
\begin{lem}
\label{lem:os}
If \typj{\Gam}{P:w} and $P\red P'$, then $\osG{\Gam}{P} >_{mul}\osG{\Gam}{P'}$.
\end{lem}

\begin{thm}[Soundness]\label{thm:soundness:io}
  If \typj{\Gam}{P:w}, then $P$ terminates.
\end{thm}
\begin{proof}
  Suppose that $P$ diverges, i.e., there is an infinite
  sequence$(P_i)_{i \in N}$, where $P_i \red P_{i+1}$, $P = P_0$.
According to Theorem~\ref{thm:subject:reduction} every $P_i$ is typable. 
Using Lemma \ref{lem:os} we have \osG{\Gam}{P_{i}} $>_{mul}$
\osG{\Gam}{P_{i+1}}
for all $i$, which yields a contradiction.
\end{proof}

\begin{rk}[{\`a} la Curry vs {\`a} la Church]
  Our system is presented {\`a} la Curry. Existing systems for
  termination~\cite{deng:sangiorgi:termination:IC,demangeon:hirschkoff:sangiorgi:concur10}
  are {\`a} la Church, while the usual presentations of
  i/o-types~\cite{DBLP:journals/mscs/PierceS96} are {\`a} la Curry.
  The latter style of presentation is better suited to address type
  inference (see Section~\ref{sec:inference}). This has however some
  technical consequences in our proofs. Most
  importantly, the measure on processes (Definition~\ref{def:measure})
  would be simpler when working {\`a} la Church, because we could
  avoid to consider all possible derivations of a given judgement.
  We are not aware of Church-style presentations of i/o-types. 

\end{rk}


\section{Expressiveness of our Type System} 
\label{sec:expr}

For the purpose of the discussions in this section, we work in a
polyadic calculus. The extension of our type system to handle
polyadicity is rather standard, and brings no particular difficulty.

\subsection{A More Flexible Handling of Levels}

Our system is strictly more expressive than the original one by Deng
and Sangiorgi~\cite{deng:sangiorgi:termination:IC}, as expressed by
the two following observations (Lemma~\ref{lem:comparison:deng} and
Example~\ref{expl:level-polymorphism}):
\begin{lem}
  \label{lem:comparison:deng}
  Any process typable according to the first type system
  of~\cite{deng:sangiorgi:termination:IC} is typable in our 
  system.
\end{lem}
\begin{proof}
  The presentation of~\cite{deng:sangiorgi:termination:IC} differs
  slightly from ours. The first system presented in that paper can be
  recast in our setting by working with the \# capability only (thus
  disallowing subtyping), and requiring type \chan k T for $a$ in the
  first premise of the rules for output, finite input and replicated
  input. 
 We write \typdeng{\Gam}{P:w} for the resulting judgement.
We establish that \typdeng{\Gam}{P:w} implies \typj{\Gam}{P:w} by
induction over the derivation of \typdeng{\Gam}{P:w}.
\end{proof}

We now present an example showing that the flexibility
brought by subtyping can be useful to ease programming. We view
replicated processes as servers, or functions.
Our example shows that it is possible in our system to invoke a server
by passing names having different levels, provided some form of
coherence (as expressed by the subtyping relation) is guaranteed. This
form of ``polymorphism on levels'' is not available in previous type
systems for termination in the $\pi$-calculus.


\begin{expl}[Level-polymorphism]
  \label{expl:level-polymorphism}
  Consider the following definitions (in addition to polyadicity, we
   accommodate the first-order type of natural numbers,
  with corresponding primitive operations):
  \begin{mathpar}
    \begin{array}{rcl}
      F_1 &=&
      !f_1(n,r).\outm{r}{n*n}\\
      F_2 &=&
    !f_2(m,r).(\new s)\,\big( \outm{f_1}{m+1,s}~|~ s(x).\outm{r}{x+1} \big)\\
      Q &=&
    !g(p,x,r).(\new s)\,\big( \outm{p}{x,s}~|~ s(y).\outm{p}{y,r} \big)
  \end{array}
  \end{mathpar}
$F_1$ is a server, running at $f_1$, that returns the square of a
integer on a continuation channel $r$ (which is its second argument).
$F_2$ is a server that computes similarly $(m+1)^2+1$, by making a
call to $F_1$ to compute $(m+1)^2$. Both $F_1$ and $F_2$ can be viewed
as implementations of functions of type \texttt{int -> int}.

$Q$ is a ``higher-order server'': its first argument $p$ is the
address of a server acting as a function of type \texttt{int -> int},
and $Q$ returns the result of calling twice the function located at
$p$ on its argument (process $Q$ thus somehow acts like Church numeral $2$).

Let us now examine how we can typecheck the process
\begin{mathpar}
  F_1~|~F_2~|~Q~|~    \outm{g}{f_1,4,t_1}~|~  \outm{g}{f_2,5,t_2}
  \enspace.
\end{mathpar}
$F_2$ contains a call to $f_1$ under a replicated input on $f_2$,
which forces $\lvl{f_2}>\lvl{f_1}$. In the type systems
of~\cite{deng:sangiorgi:termination:IC}, this prevents us from
typing the processes above, since $f_1$ and $f_2$ should have the same
type (and hence in particular the same level), both being used as
argument in the outputs on $g$.
We can type this process in our setting, thanks to subtyping, for
instance by assigning the following types:
$g:\ochan{k_g}{(\ochan{k_2}{T}, U,V) }, \, f_2:\chan{k_2}{T},\, f_1:
\chan{k_1}{T}$, with $k_1<k_2$.
\end{expl}

It can be shown that this example cannot be typed using any of the
systems of~\cite{deng:sangiorgi:termination:IC}.
It can however be phrased (and hence recognised as terminating) in the
``purely functional $\pi$-calculus''
of~\cite{demangeon:hirschkoff:sangiorgi:concur10}, that is, using a
semantics-based approach  --- see also
Section~\ref{sec:iofun}. It should however not be difficult to present
a variation on it that forces one to rely on levels-based type
systems.

\subsection{Encoding  the Simply-Typed $\lambda$-calculus} 

\newcommand{\clause}[1]{}

We now push further the investigation of the ability to analyse
terminating functional behaviour in the $\pi$-calculus using our type
system, and study an encoding of the $\lambda$-calculus in the
$\pi$-calculus. 

We focus on the following \emph{parallel call-by-value} encoding,
but we believe that the analogue of the results we present here also
holds for other encodings.
A $\lambda$-term $M$ is encoded as $\encod{M}_p$, where $p$ is a
name which acts as a parameter in the encoding. The encoding is
defined as follows:
%
\begin{mathpar}
  \clause{Lam}~~
  \encod{\lambda x.M}_p
  \eqdef
  (\new y)\,(!y(x,q).\encod{M}_q\,|\, \outm{p}{y})
  \and
  \clause{Var}~~
  \encod{x}_p
  \eqdef
  \outm p x
  \and
  \clause{App}~~
  \encod{M\,N}_p
  \eqdef
  (\new q,r)\,\big(\,
  \encod{M}_q~|~\encod{N}_r~|~
  q(f).r(z).\outm{f}{z,p}
  \,\big)
\end{mathpar}
We can make the following remarks:
\begin{itemize}
\item A simply-typed $\lambda$-term is encoded into a simply-typed
  process (see~\cite{SW01}). Typability for termination comes into
  play in the translation of $\lambda$-abstractions.
\item The target of this encoding is $L\pi$, the  \emph{localised
    $\pi$-calculus} in which only the output capability is
  transmitted (see also Section~\ref{section:inference_lpi}).
\end{itemize}


\cite{DBLP:conf/birthday/DemangeonHS09} provides a counterexample to
typability of this encoding for the first type system
of~\cite{deng:sangiorgi:termination:IC} (the proof of this result also
entails that typability according to the other, more expressive, type
systems due to Deng and Sangiorgi also fails to hold). 
Let us analyse this example:
\begin{expl}[From~\cite{DBLP:conf/birthday/DemangeonHS09}]
  The $\lambda$-term
  $
    M_1~~\eqdef~~ f ~ ( \lambda x. (f ~ u ~(u ~ v))
  $ 
can be typed in the simply typed $\lambda$-calculus, in a typing
context containing the hypotheses $f:(\sigma\rar\tau)\rar\tau\rar\tau, v:\sigma,
u:\sigma\rar\tau$.

Computing $\encod{M_1}_p$ 
yields the process:
\begin{mathpar}
  \begin{array}{l}
        (\new q, r)\\
        \quad(\new y)\,\big(~\outm r y\\
        \left.\begin{array}{ll}
        \quad\quad\quad
        |~~    !y(x, q').(\new q_1,r_1,q_2,r_2,q_3,r_3)\\
        \quad\quad\quad\quad
        \big(~~
    {\outm {q_2} f ~|~ \outm {r_2} u ~|~
      q_2(f_2).r_2(z_2).\outm{f_2}{z_2, q_1}}
    \quad~~~
    &{\encod{f~u}_{q_1}}
    \quad~
  \\[.2em]
        \quad\quad\quad\quad
  ~|~~
    {
      \outm {q_3} u ~|~ \outm {r_3} v ~|~
      q_3(f_3).r_3(z_3).\outm{f_3}{z_3,r_1}}
    & {\encod{u~v}_{r_1}}
  \\[.2em]
        \quad\quad\quad\quad
  ~|~~
   q_1(f_1).r_1(z_1).\outm{f_1}{z_1, q'} ~~\big)~\big)
 \end{array}\right]
~\encod{\lambda x.\,(f~ u~(u~v))}_r
 \\[.2em]
 \quad
 ~|~     \outm q f~|~ q(f').r(z).\outm{f'}{z, p}
\end{array}
\end{mathpar}
%

If we try and type this term using the first type system
of~\cite{deng:sangiorgi:termination:IC}, we can reason as follows:
\begin{enumerate}
\item\label{obs:one} By looking at the line corresponding to $\encod{f~u}_{q_1}$, we
  deduce that the types of $f$ and $f_2$ are unified, and similarly
  for $z_2$ and $u$.
\item\label{obs:two} Similarly, the next line ($\encod{u~v}_{r_1}$) implies that the
  types of $f_3$ and $u$ are unified.
\item\label{obs:three} The last line above entails that the types
  assigned to $f$ and $f'$ must be unified, and the same for the types
  of $z$ and $y$ (because of the output \outm r y).
  \end{enumerate}

  If we write \chan{k}{\msg{T_1,T_2}} for the type (simple) assigned to $f$,
  we have by remark~\ref{obs:one} that $u$ has type $T_1$, and the
  same holds for $y$ by remark~\ref{obs:three}.
  In order to typecheck the replicated term, we must have $\lvl y>\lvl
  {f_3} = \lvl u$ by remark~\ref{obs:two}, which is impossible since
  $y$ and $u$ have the same type.


  While $\encod{M_1}_p$ cannot be typed using the approach
  of~\cite{deng:sangiorgi:termination:IC}, it can be using the system of
  Section~\ref{section:io}. Indeed, in that setting $y$ and $u$ need
  not have the same levels, so that we can satisfy the constraint
  $\lvl y>\lvl u$. The last line above generates an output
  \outm{f}{y,p}, which can be typed directly, without use of
  subtyping. To typecheck the output \outm{f}{u,q_1}, we ``promote''
  the level of $u$ to the level of $y$ thanks to subtyping, which is
  possible because only the output capability on $u$ is transmitted
  along $f$.
%


  
\end{expl}

It however appears that our system is not able to typecheck the image
of \STlam, as the following (new) counterexample shows:
%
%
\begin{expl}\label{expl:new:lambda}
  We first look at the following rather simple $\pi$-calculus process:
$$
(\new u)\,\big(\,
!u(x).\out x ~|~ (\new v)\,(!v.\outm u t ~|~ \outm u v)
\,\big)
\enspace.
$$
%
This process is not typable in our type system, although it
terminates. Indeed, we can assign a type of the form
\chan{k}{\ochan{n}{\Unit}} to $u$, and \chan{m}{\Unit} to
$v$. Type-checking the subterm $!v.\outm u t$ imposes $k<m$, and
type-checking $!u(x).\out x$ imposes $k>n$. Finally, type-checking
\outm u v gives $m\leq n$, which leads to an
inconsistency.


%
We can somehow `expand' this process into the encoding of a
$\lambda$-term: consider indeed 
$$
M_2~~\eqdef~~
\big(\,
\lambda u.\,((\lambda v. (u~v))~(\lambda y. (u~t)))
\,\big)
~~(\lambda x.\,(x~a))
\enspace.
$$
%
We do not present the (rather complex) process corresponding to
$\encod{M_2}_p$. We instead remark that there is a sequence of
reductions starting from $\encod{M_2}_p$ and leading to
  $$
  ! y_1(u, q_1).\big(~ ! y_3(v, q_4). \outm {{ u}} {{ v}, q_4} ~|~ !y_5 (y, q_5). \outm {{ u}}{{ t}, q_5}~|~ 
  \outm {y_3}{y_5, q_1} ~\big) ~~|~~
  ! y_2(x, q_2). \outm {{ x}} {a, q_2}  ~~|~~ \outm {y_1}{y_2,p}
  \enspace.
  $$
  These first reduction steps correspond to `administrative
  reductions' (which have no counterpart in the original
  $\lambda$-calculus term). We can now perform some communications
  that correspond to $\beta$-reductions,
and obtain a process which contains a subterm of the form
  $$
    \outm{{\bf u}} {{\bf v},p}~|~ {\bf !v}(y, q_5). { \outm {{\bf u}}
      {{\bf t},q_5}} ~|~ {\bf !u(x},q_2{\bf)}.{\bf \out x} 
    \msg{a, q_2}
    \enspace.
  $$
Some channel names appear in boldface  in
order to stress the similarity with the process seen above: for the
same reasons, this term cannot be typed. 
By subject reduction (Theorem~\ref{thm:subject:reduction}), a typable
term can only reduce to a typable term. This allows us to conclude
that $\encod{M_2}_p$ is not typable in our system.

\end{expl}

\subsection{Subtyping and Functional Names}
\label{sec:iofun}
\newcommand{\curt}[3]{\ensuremath{#1\,\vdash\,#2:#3}}

\newcommand{\separator}{\ensuremath{\,\bullet\,}}
\newcommand{\compositenv}[2]{\ensuremath{#1 \separator #2}}

\newcommand{\typf}[3]{\ensuremath{\compositenv{#1}{#2}\,\vdash\, #3}}
\renewcommand{\typ}[2]{\ensuremath{#1\,\vdash\,#2}}

In order to handle functional computation as expressed by \STlam, we
extend the system of Section~\ref{section:io} along the lines
of~\cite{demangeon:hirschkoff:sangiorgi:concur10}.
The idea is to classify names into \emph{functional} and
\emph{imperative} names. Intuitively, functional names arise through
the encoding of \STlam{}. For termination, these are dealt with using
an appropriate method --- the `semantics-based' approaches discussed
in Section~\ref{sec:intro}, and introduced
in~\cite{yoshida:berger:honda:termination:ic,sangiorgi:mscs:termination}. For
imperative names, we resort to (an adaptation of) the rules of
Section~\ref{section:io}.


Our type system is {\`a} la Curry, and the kind of a name, functional
or imperative, is fixed along the construction of a typing derivation.
Typing environments are of the form 
\compositenv{\Gam}{f:\ochan k T} --- the intuition is that we
{isolate} a particular name, $f$. $f$ 
is the name which can be used to build replicated inputs where $f$ is
treated as a functional name.  The typing rules are given on
Figure~\ref{fig:typing:iofun}.  There are two rules to typecheck a
restricted process, according to whether we want to treat the
restricted name as functional (in which case the isolated name
changes) or imperative (in which case the typing hypothesis is added
to the \Gam\ part of the typing environment).

\begin{figure}[t]
  \begin{mathpar}
  \inferrule*{\typf{\Gam, x:T}{-}{P:w}
    \and k\geq w
  }{
    \typf{\Gam}{f:\ochan k T}{!f(x).P:0}
  }
  \and
  \inferrule*{\typ{\Gam, f:\ochan k T}{a:\ochan n U}
    \and
    \typ{\Gam, f:\ochan k T}{v:U}
  }{
    \typf{\Gam}{f:\ochan k T}{\outm{a}{v}:n}
  }
  \and
  \inferrule*{
    \typ{\Gam}{c:\ichan n T}\and
    \typf{\Gam, x:T, f:\ochan k U}{-}{P:w}\and
    n>w
  }{
    \typf{\Gam}{f:\ochan k U}{c(x).P:0}
  }
  \and
  \inferrule*{
    \typ{\Gam}{c:\ichan{n} T}\and
    \typf{\Gam, x:T, f:\ochan k U}{-}{P:w}\and
    n>w
  }{
    \typf{\Gam}{f:\ochan k U}{!c(x).P:0}
  }
  \and
  \inferrule*{
    \typf{\Gam}{f:\ochan k T}{P_1}
    \and
    \typf{\Gam}{f:\ochan k T}{P_2}
  }{
    \typf{\Gam}{f:\ochan k T}{P_1|P_2}
    }
  \and
  \inferrule*{\typf{\Gam, g:\ochan{k}{T}}{f:\ochan{n}{U}}{P:w}
  }{
    \typf{\Gam}{g:\ochan k T}{(\new f)\,P:w}
  }
  \and
  \inferrule*{
    \typf{\Gam, c:\chan n T}{f:\ochan k U}{P:w}
  }{
    \typf{\Gam}{f:\ochan k U}{(\new c)\,P:w}
  }
\end{mathpar}
  \caption{Typing Rules for an Impure Calculus}
\label{fig:typing:iofun}
\end{figure}

\medskip

The typing rules of Figure~\ref{fig:typing:iofun} rely on
i/o-capabilities and the isolated name to enforce the usage of
functional names as expressed in~\cite{sangiorgi:mscs:termination}.
In~\cite{demangeon:hirschkoff:sangiorgi:concur10}, a specific
syntactical construct 
is instead used: we manipulate processes of the form
$\mathtt{def~}f\mathtt{~=~}(x)P_1 \mathtt{~in~}P_2$ (that can be read
as $(\new f)\,(!f(x).P_1\,|\,P_2)$), where $f$ does not occur in $P_1$
and occurs in output position only in $P_2$.

Let us analyse how our system imposes these constraints.  In the rule
for restriction on a functional name, 
the name $g$, that occurs in `isolated position' in the conclusion of
the rule, is added in the `non isolated' part of the typing
environment in the premise, with a type allowing one to use it in
output only.

In the rules for input on an imperative name (replicated or not), the
typing environment is of the form \compositenv{\Gam}{-} in the premise
where we typecheck the continuation process: this has to be understood
as \compositenv{\Gam}{d:\ochan k T}, for some dummy name $d$ that is
not used in the process being typed. We write `$-$' to stress the fact
that we disallow the construction of replicated inputs on functional
names. The functional name $f$ appears in the aforementioned premise in
the `non isolated' part of the typing environment, with only the
output rights on it.
Forbidding the creation of replicated inputs on functional names under
input prefixes is necessary because of diverging terms like the
following ($c$ is imperative, $f$ is functional):
$$
c(x).!f(y).\outm x y~~|~ \outm c f~|~ \outm f v
\enspace.
$$

Note also that typing non replicated inputs (on imperative names)
involves the same constraints as for replicated inputs, like
in~\cite{demangeon:hirschkoff:sangiorgi:concur10}: the relaxed control
over functional names requires indeed to be more restrictive on all
usages of imperative names.

The notation \compositenv{\Gam}{-} is also used in the rule to type a
replicated input on a functional name, and we can notice that in this
case $f$ cannot be used at all in the premise, to avoid
recursion.

\medskip


In addition to the gain in expressiveness brought by subtyping, we can
make the following remark:
\begin{rk}[Expressiveness]
  As in~\cite{demangeon:hirschkoff:sangiorgi:concur10}, our system
  allows one to typecheck the encoding of a \STlam{} term, by treating
  all names as functional, and assigning them level $0$.

  Moreover, our type system makes it possible to typecheck processes
  where several replicated inputs on the same functional name coexist,
  provided they occur `at the same level' in the term. For instance, a
  term of the form $(\new f)\,(!f(x).P~|~!f(y).Q~|~R)$ can be
  well-typed with $f$ acting as a functional name. This is not
  possible using the \texttt{def} construct
  of~\cite{demangeon:hirschkoff:sangiorgi:concur10}.
\end{rk}
Another form of expressiveness brought by our system is given by
typability of the following process: $!u(x).\out x~|~ !v.\outm u t~|~
\outm u v~|~ c(y).\outm u c$. Here, name $c$ must be imperative while
name $v$ must be functional, and both are emitted on $u$. This is
impossible in~\cite{demangeon:hirschkoff:sangiorgi:concur10}, where
every channel carries either a functional or an imperative name. In
our setting, only the output capability on $c$ is transmitted along
$u$, so in a sense $c$ is transmitted `as a functional name'.

Because of the particular handling of restrictions on functional
names, the analogue of Lemma~\ref{lem:subject:congruence} does not
hold for this type system: typability is not preserved by structural
congruence. Accordingly, the subject reduction property is stated in
the following way:
\begin{thm}[Subject reduction]
  If \typf{\Gam}{f:\ochan k T}{P:w} and $P\red P'$, then there exist
  $Q$ and $w'\leq w$ s.t.\ $P'\equiv Q$ and \typf{\Gam}{f:\ochan k
    T}{Q:w'}.
\end{thm}

\begin{thm}[Soundness]
  If \typf{\Gam}{f:\ochan k T}{P:w}, then $P$ terminates.
\end{thm}
\begin{proofsketch}
  The proof has the same structure as the corresponding proof
  in~\cite{demangeon:hirschkoff:sangiorgi:concur10}. An important
  aspect of that proof is that we exploit the termination property for
  the calculus where all names are functional without looking into
  it. To handle the imperative part, we must adapt the proof along the
  lines of the termination argument for
  Theorem~\ref{thm:soundness:io}.
\end{proofsketch}


\section{Type Inference}
\label{sec:inference}
We now study type inference, that is, given a process $P$, the
existence of \Gam{}, $w$ such that \typj{\Gam}{P:w}. There might a
priori be several such \Gam{} (and several $w$: see
Lemma~\ref{lem:bigger_weight}). Type inference for level-based systems
has been studied in~\cite{DBLP:conf/tgc/DemangeonHKS07}, in absence of
i/o-types. 
We first present a type inference procedure in a special case of our
type system, and then discuss this question in the general case.


\subsection{Type Inference for Termination in the Localised
  $\pi$-calculus} 
\label{section:inference_lpi}

In this section, we 
concentrate on the \emph{localised
  $\pi$-calculus}, \Lpi, which is defined by imposing that channels
transmit only the output capability on names: a process like
$a(x).x(y).\nil$ does not belong to \Lpi, as it makes use of the input
capability on $x$.
From the point of view of implementations, the restriction to \Lpi{}
makes sense. For instance, the 
language JoCaml~\cite{jocaml} implements a variant of the
$\pi$-calculus that follows this approach: one can only use a received
name in output.
Similarly, the communication primitives in
Erlang~\cite{erlang:website} can also be viewed as obeying to the
discipline of \Lpi: asynchronous messages can be sent to a PiD
(process id), and one cannot create dynamically a receiving agent at
that PiD: the code for the receiver starts running as soon as the PiD
is allocated. 

Technically, \Lpi{} is introduced by allowing the transmission of
o-types only. We write \typlpi{\Gam}{P:w} if \typj{\Gam}{P:w} can be
derived in such a way that in the derivation, whenever a type of the
form $\eta^k\eta'^{k'}T$ occurs, we have $\eta'=\otypename$ (types of
the form \ichan{k}{T} and \chan{k}{T} appear only when typechecking
input prefixes and restrictions).
Obviously, typability for \typlpi{}{} entails typability for
\typj{}{}, hence termination.
It can also be remarked that in restricting to \Lpi{}, we keep
an important aspect of the flexibility
brought by our system. In particular, the examples we have discussed
in Section~\ref{sec:expr} --- Example~\ref{expl:level-polymorphism}, and
the encoding of the $\lambda$-calculus --- belong to \Lpi.

\medskip

We now describe a type inference procedure for \typlpi{}{}.
For lack of space, we do not provide
all details and proofs.


We first check typability when levels are not taken into
account. For this, we
rely on a type inference algorithm for simple
types~\cite{DBLP:conf/concur/VasconcelosH93}, together with a simple
syntactical check to verify that no received name is used in
input. When this first step succeeds, we
replace \chan{}T types with \ochan{}T types appropriately in the
outcome of the procedure for simple types (a type variable may be
assigned to some names, as, e.g., to name $x$ in process $a(x).\outm b
x$).
%
%


What remains to be done is to find out whether types can be decorated
with levels in order to ensure termination.
%
As mentioned above,
we suppose w.l.o.g.\ that we
have a term $P$ in which all bound names are pairwise distinct, and
distinct from all free names. We define the following
sets of names:
\begin{itemize}
\item \names{P} stands for the set of all names, free and bound, of $P$;
\item \bnames{P} is the set of names that appear bound (either by
  restriction or by input) in $P$;
\item \rcvnames{P} is the set of names that are bound by an input
  prefix in $P$ ($x\in\rcvnames P$ iff $P$ has a subterm of the form
  $a(x).Q$ or $!a(x).Q$ for some $a, Q$);
\item \resnames{P} stands for the set of names that are restricted in $P$
  ($a\in\resnames P$ iff $P$ has a subterm of the form $(\new a)\,Q$
  for some $Q$).
\end{itemize}
We have $\bnames P = \rcvnames P\uplus \resnames P$ (where $\uplus$
stands for disjoint union), and $\names P = \bnames P \uplus \fnames
P$. Moreover, for any $x\in\rcvnames P$, there exists a unique
$a\in\fnames P\cup \resnames P$ such that $P$ contains the prefix
$a(x)$ or the prefix $!a(x)$: we write in this case $a = \father{x}$
($a\in\fnames P\cup\resnames P$, because we are in \Lpi).

We 
build a graph as follows:
\begin{itemize}
\item  For every name $n\in\fnames P \cup \resnames P$, create a node
  labelled by $n$, and create a node labelled by \son{n}. Intuitively,
  if $n$ has type $\chan k S$ of \ochan k S, \son n{} has type $S$.
In case type inference for simple types returns a type of the form
$\alpha$, where $\alpha$ is a type variable, for $n$, we just create the
node $n$.
\item For every $x\in\rcvnames P$, let $a=\father{x}$, add $x$ as a label to
  \son{a}.
\end{itemize}

\begin{expl}\label{expl:inference:one}
  We associate to the process $P= a(x).(\new b)\,\outm x b~\,|\,~
  !a(y).(\outm c y~|~ d(z).\outm y z)$ the following set of 8 nodes with
  their labels:
  $\set{a}, \set{\son a, x, y}, \set{b}, \set{\son b}, \set{c},
  \set{\son c, y}, \set{d}, \set{\son d, z}$.
\end{expl}

The next step is to insert edges in our graph, to represent the
constraints between levels. 
\begin{itemize}
\item For every output of the form \outm n m, we insert an edge
  labelled with ``$\geq$'' from \son{n} to $m$.
\item For every subterm of $P$ of the form $!a(x).Q$, and for every
  output of the form \outm n m that occurs in $Q$ without occurring
  under a replication in $Q$, we insert an edge $a\xr{>}n$.
\end{itemize}

\begin{expl}\label{expl:inference:graph}
  The graph associated to process ~$!c(z).\outm b z~|~
  \outm a c~|~ \outm a b$~  
has nodes $$\set{a}, \set{\son a}, \set{b}, \set{\son
    b}, \set{c}, \set{\son c, z}\enspace,$$ and can be depicted as follows:
%
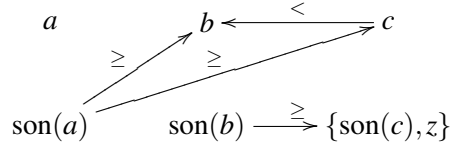
\qquad
$
         \xymatrix{ a & b  & \ar[l]_{<} c\\
                    \son a \ar[ur]^{\geq} \ar[urr]^{\geq} & \son b
                    \ar[r]^{\geq} & 
                     \set{\son c, z}} 
$
\end{expl}

The last phase of the type inference procedure consists in looking for
an  assignment of levels on the graph: this is possible as long as there are no
cycles involving at least one \xr{>} edge in the graph.

\newcommand{\setnodes}{\ensuremath{\mathcal S}}

At the beginning, all nodes of the graph are unlabelled; we shall label
them using natural numbers.
\begin{enumerate}
\item\label{inference:one} We go through all nodes of the graph, and
  collect those that have no outgoing edge leading to an unlabelled
  node in a set \setnodes.

\item\label{inference:two} If \setnodes{} is not empty, we label every
  node $n$ in \setnodes{} as follows: we start by setting $n$'s label
  to $0$. 

  We then examine all outgoing edges of $n$. For every
  $n\xr{\geq}m$, we replace $n$'s label, say $k$, with \max{k,k'},
  where $k'$ is $m$'s label, and similarly for $n\xr{>}m$ edges, with
  \max{k,k'+1}.

%
  We then empty \setnodes, and start again at step~\ref{inference:one}.

\item If $\setnodes=\emptyset$, then either all nodes of the graph are
  labelled, in which case the procedure terminates, or the graph
  contains at least one oriented cycle. If this cycle contains at
  least one \xr{>} edge, the procedure stops and reports
  failure. Otherwise, the cycle involves only \xr{\geq} edges: 
  we compute the level of each node of the cycle along the lines of
  step~\ref{inference:two} (not taking into account nodes of the cycle
  among outgoing edges), and then assign the maximum of these labels
  to all nodes in the cycle.
  We start again at
  step~\ref{inference:one}.

\end{enumerate}




This procedure terminates, since each time we go back to
step~\ref{inference:one}, strictly more nodes are labelled.

\begin{expl}
  On the graph of Example~\ref{expl:inference:graph}, the procedure
  first assigns level $0$ to nodes $a, b$ and \set{\son c, z}. In the
  second iteration, $\setnodes = \set{\son b, c}$; level $0$ is
  assigned to \son b, and $1$ to $c$.  
  Finally, level $1$ is assigned to \son a.
%
%
This yields the typing $b:\ochan 0 \ochan 0 T,
c: \chan 1 \ochan 0 T, a:\ochan 0 \ochan 1 \ochan 0 T$ for the
process of Example~\ref{expl:inference:graph}.
\end{expl}

As announced above, for lack of space we have described only the main
steps of our type inference procedure. Establishing that the latter
has the desired properties involves the introduction of an auxiliary
typing judgement (that characterises \typlpi{}{}), and explaining how
types are reconstructed at the end of the procedure. This finally
leads to the following result:
\begin{thm}
  There is a type inference procedure that given a process $P$,
  returns $\Gam, w$ s.t.\   \typlpi{\Gam}{P:w} iff there exists
  $\Gam', w'$ s.t.\   \typlpi{\Gam'}{P:w'}.


\end{thm}

\subsection{Discussion: Inferring i/o-Types}

If we consider type inference for the whole system of
Section~\ref{section:io}, the situation is more complex. We start by
discussing type inference without taking the levels into account.
If a process is typable using simple types
(that is, with only types of the form \chan{}T), one is interested in
providing a more informative typing derivation, where input and output
capabilities are used.

For instance, the process $a(x).\outm x t$ can be typed using
different assignments for $a$: \ichan{}{\ochan{}{T}},
\chan{}{\ochan{}T}, \ichan{}{\chan{}T}, and \chan{}{\chan{}T} --- if
we suppose $t:T$.  Among these, \ichan{}{\ochan{}T} is the most
informative (intuitively, types featuring `less \#' seem preferable
because they are more precise).
Moreover, it is a supertype of all other types, thus acting as a
`candidate' if we were to look for a notion of principal typing.
%
%
Actually, in order to infer i/o-types, one must be able to compute
lubs and glbs of types, using equations like $glb(\ichan{}T,\ichan{}U)
= \ichan{}{\,glb(T,U)}$, $glb(\ichan{}T, \ochan{}U) =
\chan{}{glb(T,U)}$, and $glb(\ochan{}T, \ochan{}U) =
\ochan{}{\,lub(T,U)}$. The contravariance of \ochan{}{} suggests the
introduction of an additional capability, that we shall note
\tchan{}{}, which builds a supertype of input and output capabilities
(more formally, we add the axioms $\ichan{}T \leq \tchan{}T$ and
$\ochan{}T \leq \tchan{}T$).

\cite{DBLP:journals/iandc/IgarashiK00} presents a type inference
algorithm for (an enrichment of) i/o-types, where such a capability
\tchan{}{} is added to the system
of~\cite{DBLP:journals/mscs/PierceS96} (the notations are different,
but we adapt them to our setting for the sake of readability).
The use of \tchan{}{} can be illustrated on the following example
process:
$$
Q_1 \eqdef a(t).b(u).\big(~ !t(z).\outm u z~|~ \outm c t~|~ \outm c u
~\big)
\enspace.
$$
To typecheck $Q_1$, we can see that the input (resp.\ output)
capability on $t$ (resp.\ $u$) needs to be received on $a$ (resp.\
$b$), which suggests the types $a:\ichan{}{\ichan{}T},
b:\ichan{}{\ochan{}T}$. Since $t$ and $u$ are emitted on the same
channel $c$, and because of contravariance of output, we compute
a \emph{supertype} of \ichan{}{T} and \ochan{}{T}, and 
assign type \ochan{}{\tchan{}{T}} to $c$.

Operationally, the meaning of \tchan{}{} is ``no i/o-capability at
all'' (note that this does not prevent from comparing names, which may
be useful to study behavioural
equivalences~\cite{DBLP:journals/mscs/HennessyR04}): 
in the typing we just described, since we only have the input
capability on $t$ and the output capability on $u$, we must renounce
to all capabilities, and $t$ and $u$ are sent without the receiver to
be able to do anything with the name except passing it along.
Observe also that depending on how the context uses $c$, a different
typing can be introduced. For instance, $Q_1$ can be typed by setting
$a:\ichan{}{\chan{}T}, b:\ichan{}{\ochan{}T},
c:\ochan{}{\ochan{}T}$. This typing means that the output capability
on $u$ is received, used, and transmitted on $c$, and both
capabilities on $t$ are received, the input capability being used
locally, while the output capability is transmitted on $c$.

The first typing, which involves \tchan{}{}, is the one that is
computed by the procedure
of~\cite{DBLP:journals/iandc/IgarashiK00}. 
It is ``minimal'', in the terminology
of~\cite{DBLP:journals/iandc/IgarashiK00}.  Depending on the
situations, a typing like the second one (or the symmetrical case,
where the input capability is transmitted on $c$) might be
preferable. 

\medskip

If we take levels into account, and try and typecheck $Q_1$ (which
contains a replicated subterm), the typings mentioned above can be
adapted as follows: we can set $a:\ichan 0 \chan 1 T, b:\ichan 0\ochan
0 T, c:\ochan 0 {\ochan 1 T}$, 
in which case subtyping on levels is used to deduce
$u:\ochan 1T$ in order to typecheck \outm c u.
Symmetrically, we can also set $a: \ichan 0\ichan 1 T, b:\ichan
0\chan 0T, c:\ochan 0 {\ichan 0 T}$, and
typecheck \outm c t{} using subsumption to deduce $t: \ichan 0 T$.

It is not clear to us how levels should be handled in relation with
the \tchan{}{} capability. One could think that since \tchan{}{}
prevents any capability to be used on a name, levels have no use, and
one could simply adopt the subtyping axioms $\ichan k T \leq
\tchan{}T$ and $\ochan k T\leq \tchan{}T$. This would indeed allow us
to typecheck $Q_1$.

Further investigations on a system for i/o-types with \tchan{}{} and
levels is left for future work, as well as the study of inference for
such a system.



\section{Concluding Remarks}
\label{sec:concl}

In this paper, we have demonstrated how Pierce and Sangiorgi's
i/o-types can be exploited to refine the analysis of the
simplest of type systems for termination of processes presented
in~\cite{deng:sangiorgi:termination:IC}. Other, more complex systems
are presented in that work, and it would be interesting to study
whether they would benefit from the enrichment with capabilities and
subtyping.
One could also probably refine the system of Section~\ref{section:io}
by distinguishing between \emph{linear} and \emph{replicated input
  capabilities}, as only the latter must be controlled for
termination (if a name is used in linear input only, its level is
irrelevant). 


The question of type inference for our type systems (differently from
existing proposals, these are presented {\`a} la Curry, which is
better suited for the study of type inference) can be studied further. 
It would be interesting to analyse how the procedure of
Section~\ref{section:inference_lpi} could be ported to programming
languages that obey the discipline of \Lpi{} for communication, like
Erlang or JoCaml.
For the moment, we only have preliminary results for a type inference
procedure for the system of Section~\ref{section:io}, and we would
like to explore this further. 
Type inference for the system of Section~\ref{sec:iofun} is a
challenging question, essentially because making the distinction
between functional and imperative names belongs to the inference
process (contrarily to the setting
of~\cite{demangeon:hirschkoff:sangiorgi:concur10}, where the syntax of
processes contains this information).


\paragraph{Acknowledgements.} Romain Demangeon, as well as anonymous
referees, have provided insightful
comments and suggestions on this work. We also acknowledge support by
ANR projects ANR-08-BLANC-0211-01 "COMPLICE",
ANR-2010-BLANC-0305-02 "PiCoq" and CNRS PEPS "COGIP".

{\small
\bibliographystyle{eptcs}

}

\end{document}

\textbf{Todo}
{\tiny
{discuss related works, in particular other systems with
  subtyping in pi}
/
{check that the typing rules obey the conventions adopted above
  (for instance, ``new c'' or ``new a''?); similarly check that \outm
  a v is used, and not \outm a b, in typing rules}
/ {n'y a-t-il pas un peu trop d'auto-citation?}
%
%
/
verifier qu'il y a assez d'explications/exemples dans les parties
introductives
/
en iflong, je vire le nom des regles: verifier que c'est ok
/
{some references:
  \cite{
    DBLP:conf/birthday/PierceT00,DBLP:journals/iandc/IgarashiK00}
%
/ possibly mettre la def. de la mesure dans une figure, pour gagner de
la place
}
}

\end{document}
\newpage
\appendix
\input{algo-inf}

\end{document}